\documentclass[11pt,a4paper]{article}

\usepackage[T1]{fontenc}
\usepackage[utf8]{inputenc}
\usepackage{lmodern}
\usepackage{amsmath,amssymb,amsthm,mathtools}
\usepackage[margin=30mm]{geometry}
\usepackage{microtype}
\usepackage{hyperref}

\hypersetup{
  hidelinks,
  pdftitle={A note on a problem in communication complexity},
  pdfauthor={Henning Wunderlich},
  pdfsubject={Communication complexity, probabilistic operators, and margin complexity},
  pdfkeywords={communication complexity; counting protocols; probabilistic operators;
    margin complexity; threshold representations; MSC 2020: 68Q11, 68Q15, 68Q32},
  pdfdisplaydoctitle=true,
  bookmarksnumbered=true
}

\newtheorem{theorem}{Theorem}[section]
\newtheorem{proposition}[theorem]{Proposition}
\newtheorem{lemma}[theorem]{Lemma}
\newtheorem{corollary}[theorem]{Corollary}
\theoremstyle{definition}
\newtheorem{definition}[theorem]{Definition}
\theoremstyle{remark}
\newtheorem{remark}[theorem]{Remark}

\newcommand{\B}{\{0,1\}}
\newcommand{\Pcc}{\mathsf{P}^{\mathrm{cc}}}
\newcommand{\PPcc}{\mathsf{PP}^{\mathrm{cc}}}
\newcommand{\UPPcc}{\mathsf{UPP}^{\mathrm{cc}}}
\newcommand{\PHcc}{\mathsf{PH}^{\mathrm{cc}}}
\newcommand{\PSPACEcc}{\mathsf{PSPACE}^{\mathrm{cc}}}
\newcommand{\ParityPcc}{\mathord{\oplus}\mathsf{P}^{\mathrm{cc}}}
\newcommand{\SharpPcc}{\mathord{\#}\mathsf{P}^{\mathrm{cc}}}
\newcommand{\BPdot}{\mathsf{BP}\mathbin{\cdot}}
\newcommand{\RPdot}{\mathsf{RP}\mathbin{\cdot}}
\newcommand{\ZPdot}{\mathsf{ZP}\mathbin{\cdot}}
\newcommand{\PPdot}{\mathsf{PP}\mathbin{\cdot}}
\newcommand{\polylog}{\operatorname{polylog}}
\newcommand{\IP}{\operatorname{IP}}
\newcommand{\Maj}{\operatorname{MAJ}}
\newcommand{\acc}{\operatorname{acc}}
\newcommand{\rej}{\operatorname{rej}}
\newcommand{\gap}{\operatorname{gap}}
\newcommand{\PPcost}{\operatorname{PP}}
\newcommand{\UPPcost}{\operatorname{UPP}}
\newcommand{\disc}{\operatorname{disc}}
\newcommand{\mc}{\operatorname{mc}}
\newcommand{\sgn}{\operatorname{sgn}}
\newcommand{\indic}[1]{\mathbf{1}\!\left[#1\right]}
\newcommand{\normone}[1]{\left\lVert #1\right\rVert_1}
\newcommand{\RPP}{R^{\mathrm{PP}}}

\DeclareMathOperator{\co}{co}

\numberwithin{equation}{section}
\title{A note on a problem in communication complexity}
\author{Henning Wunderlich\thanks{Offenbach, Germany.
E-mail: \href{mailto:HenningWunderlich@t-online.de}{\texttt{HenningWunderlich@t-online.de}}}}
\date{\today}

\begin{document}
\maketitle

\begin{abstract}
We study bounded-error randomization and majority operators in structural
communication complexity. We give a communication-complexity formulation of
Tarui's theorem and identify the bounded-error closure of weakly
unbounded-error communication with a distributional perturbation of logarithmic
margin complexity. Retaining the Boolean guarantee in Tarui's probabilistic
polynomials yields the stronger inclusion
$\PHcc\subseteq\ZPdot\PPcc\subseteq\BPdot\PPcc$.
We also define a majority operator satisfying
$\PPdot\Pcc=\PPcc$ and use Sasaki's majority formula to prove that
$\ParityPcc\subseteq\PPdot\PPcc$.
For the inner-product function modulo two, both the outer index and each inner
protocol have logarithmic cost. Together with Forster's lower bound, this gives
an explicit witness to the non-idempotence of the majority operator and to
$\PPdot\PPcc\not\subseteq\UPPcc$.
Both the bounded-error closure and the two-level majority class are
contained in $\PSPACEcc$.
\end{abstract}



\section{Introduction}
\label{sec:introduction}

The structural theory of communication complexity studies analogues of
classical complexity classes when communication, rather than local
computation, is the resource of interest. The framework introduced by
Babai, Frankl, and Simon~\cite{BFS86} gives a particularly useful setting for
comparing nondeterminism, randomization, alternation, and counting.
Despite many unconditional separations between individual communication
classes, separating the polynomial hierarchy $\PHcc$ from
$\PSPACEcc$ remains a central challenge. We refer to
Kushilevitz and Nisan~\cite{KN97} for background and to
G\"o\"os, Pitassi, and Watson~\cite{GPW18} for a systematic account of
the structural landscape.

One approach to this challenge replaces alternating protocols by algebraic
representations and then seeks lower bounds against those representations.
Razborov~\cite{Raz89} related the communication polynomial hierarchy to
matrix rigidity over finite fields. Lokam~\cite{Lok01} developed connections
with real-valued approximations, using the probabilistic-polynomial method
of Tarui~\cite{Tar93}. Linial and Shraibman~\cite{LS09} subsequently
identified margin-complexity rigidity as another relevant parameter.
The operator viewpoint makes the common structure of these arguments
explicit: a lower bound for an appropriate bounded-error perturbation of
a counting measure becomes a lower bound for the polynomial hierarchy.
The general randomized-counting, bounded-error-operator, and distributional
perturbation framework was developed in~\cite{Wun12}, with modular
counting and rank as its principal instances.

The first part of this paper gives the PP/margin counterpart of that
framework. We work with charged counting protocols for weakly
unbounded-error communication and prove
\begin{equation}
\label{eq:intro-sandwich}
\PHcc\subseteq\ZPdot\PPcc
\subseteq\BPdot\PPcc\subseteq\PSPACEcc.
\end{equation}
Here $\mathsf{BP}$ is the bounded-error operator, and
$\mathsf{ZP}$ denotes the intersection of the two one-sided-error
closures. The bounded-error containment is the communication analogue of
Tarui's theorem. The one-sided refinement follows by retaining the
\emph{Boolean guarantee} in Tarui's construction: whenever the
approximating polynomial takes a Boolean value, that value is correct.
We include the relevant resource bounds and explain why testing equality
to either Boolean value is possible in $\PPcc$.

We also characterize $\BPdot\PPcc$ by a distributional analogue of
margin-complexity rigidity. The normalization is important: the
communication measure is comparable to the \emph{logarithm} of margin
complexity, not to margin complexity itself. The resulting class lies
between $\PHcc$ and $\PSPACEcc$. This excludes its being a strict
superclass of $\PSPACEcc$, but does not exclude equality.

The bounded-error closure is idempotent:
$\BPdot(\BPdot\PPcc)=\BPdot\PPcc$.
It is nevertheless a proper extension of $\PPcc$. Indeed, the
explicit depth-three function of Razborov and Sherstov~\cite{RS10},
together with the one-sided simulation above, gives
$\PPcc\subsetneq\ZPdot\PPcc$.
The same example separates exact logarithmic margin complexity from
its distributional perturbation.
Finite minimax also relates this perturbation to the probabilistic
margin complexity of Kamath, Montasser, and Srebro~\cite{KMS20}.
We record the resulting finite-class separation as a learning-theoretic
consequence of the communication bounds.

The second part concerns an unbounded-error majority operator with a
polylogarithmic common index. We show
\begin{equation}
\label{eq:intro-majority}
\PPdot\Pcc=\PPcc,
\qquad
\ParityPcc\subseteq\PPdot\PPcc\subseteq\PSPACEcc.
\end{equation}
The parity containment has a short proof based on Sasaki's cyclic-majority
representation of parity, recorded by Knuth~\cite[Section~7.1.1]{Knu11}.
Applied to
$\IP_n(x,y)=\sum_{i=1}^n x_i y_i\pmod 2$, it gives two levels of
majority with logarithmic resource bounds. Since
$\IP\notin\UPPcc$ by Forster's theorem~\cite{For02}, it follows that
\begin{equation}
\label{eq:intro-strict}
\PPcc\subsetneq\PPdot\PPcc,
\qquad
\PPdot\PPcc\not\subseteq\UPPcc.
\end{equation}
The underlying Boolean identity is classical; the point here is to give
an explicit operator formulation with the communication costs and
padding conventions made precise.

More generally, every Boolean predicate of an efficient accepting
count belongs to $\PPdot\PPcc$. The number of changes in the
predicate controls both the size of its outer majority representation
and a direct one-level PP upper bound. For symmetric set-intersection
functions, Sherstov's lower bound~\cite{She11} then gives a complete
class-level characterization: membership in either $\PPcc$ or
$\UPPcc$ is equivalent to having only polylogarithmically many
changes. We also quantify the outer advantage needed to turn a
two-level majority representation into bounded-error randomization.
These consequences use classical Boolean representations and known
lower bounds; no new sign-rank lower bound is claimed.

Throughout, $\PPcc$ denotes \emph{weakly} unbounded-error communication.
It must not be confused with $\UPPcc$, where the acceptance bias is
unrestricted and private randomness is used. These classes are distinct:
Buhrman, Vereshchagin, and de Wolf~\cite{BVW07} and, independently,
Sherstov~\cite{She08} proved exponential separations. Hatami, Hosseini,
and Lovett~\cite{HHL22} obtained the sharp contrast between constant
unbounded-error complexity and linear weakly unbounded-error complexity.
The dot operator in~\eqref{eq:intro-majority} is a majority quantifier,
not an unspecified oracle construction.

\section{Protocols, counting, and majority}
\label{sec:preliminaries}

\subsection{Conventions}

Let $\mathbb B^{**}=\{(x,y):x,y\in\B^*,\ |x|=|y|\}$.
A communication language is a subset $L\subseteq\mathbb B^{**}$,
and $L_n$ denotes its Boolean characteristic function on
$\B^n\times\B^n$. Alice receives $x$ and Bob receives $y$.
The protocols are nonuniform, and local computation is unrestricted.
All logarithms are to base two unless indicated otherwise.
The notation $\polylog(n)$ means $(\log(n+2))^{O(1)}$.

A protocol family is \emph{efficient} if its communication cost is
polylogarithmic in the input length. The classes $\Pcc$, $\PHcc$,
and $\PSPACEcc$ have their standard meanings~\cite{BFS86,GPW18}:
deterministic protocols, alternating protocols with a bounded number of
alternations, and alternating protocols with no restriction on the number
of alternations, respectively, all with polylogarithmic total cost.
A rectangle is a set $A\times B$, and its indicator has a deterministic
protocol of cost at most two. We take terminal outputs to be known to both
parties; changing this convention costs at most one bit.

Boolean circuits use unbounded-fan-in AND and OR gates and fan-in-one
NOT gates unless stated otherwise; size counts gates and wires.
We use two standard circuit interpretations. A language in $\PHcc$
can be represented by constant-depth, quasipolynomial-size Boolean
circuits whose inputs are rectangle indicators. A bounded-fan-in Boolean
circuit of polylogarithmic depth over such indicators, in turn, can be
evaluated by an efficient alternating protocol, and hence defines a
language in $\PSPACEcc$. Negations can be pushed to the inputs by
keeping both polarities of each gate, without more than a constant
increase in depth. A rectangle indicator and its complement both have
constant-cost deterministic protocols. At an OR gate the protocol chooses
a child existentially, and at an AND gate it chooses a child universally.
Only one root-to-leaf path is followed.

\subsection{Guess protocols}

A \emph{guess protocol} is a nonempty finite list
$\Pi=(\Pi_1,\ldots,\Pi_\ell)$ of deterministic Boolean protocols.
Put
\[
\acc_\Pi(x,y)=\sum_{i=1}^{\ell}\Pi_i(x,y),
\qquad
\rej_\Pi(x,y)=\ell-\acc_\Pi(x,y),
\]
and define $\gap_\Pi=\acc_\Pi-\rej_\Pi$.
Its cost is
\begin{equation}
\label{eq:guess-cost}
\operatorname{cost}(\Pi)
=\lceil\log\ell\rceil+\max_{i\in[\ell]}D(\Pi_i).
\end{equation}
In \emph{PP acceptance mode}, the output is one exactly when
$\gap_\Pi>0$; ties are rejected. We write $\PPcost(f)$ for the
minimum cost of a guess protocol computing $f$ in this mode and set
\[
\PPcc=\{L:\PPcost(L_n)\le\polylog(n)\}.
\]
This is a charged-counting presentation of the usual weakly
unbounded-error class~\cite{BFS86,DKMW04,Kla03,GPW18}.
The term $\log\ell$ in~\eqref{eq:guess-cost} accounts for the
precision of the majority decision.

We write $\UPPcost(f)$ for the minimum communication of a
private-coin protocol that is correct with probability strictly greater
than $1/2$ on every input; no lower bound on the advantage is charged.
Thus $\UPPcc=\{L:\UPPcost(L_n)\le\polylog(n)\}$.
Eliminating zero gaps in a PP list, choosing its index privately, and
communicating that index gives
\begin{equation}
\label{eq:upp-pp}
\UPPcost(f)\le\PPcost(f)+O(1),
\qquad \PPcc\subseteq\UPPcc.
\end{equation}

The function class $\SharpPcc$ consists of accepting-count families
of efficient guess protocols. Replacing the acceptance condition by
``$\acc_\Pi$ is odd'' defines $\ParityPcc$.
In particular, testing $x_i y_i$ on guess $i\in[n]$ shows that
$\IP\in\ParityPcc$.
The counting arithmetic and acceptance-mode conventions used below
have antecedents in~\cite{DKMW04,Wun12}. We give the elementary
constructions explicitly because their resource bounds will be used
when degrees and list lengths depend on the input size.

Two elementary normalizations will be used repeatedly. First, replacing
$\Pi$ by two copies of its list followed by one always-rejecting
protocol changes its gap from $g$ to $2g-1$. It preserves PP acceptance,
eliminates zero gaps, and increases cost by at most two.
Complementing every output thereafter negates the gap. Thus
$\PPcc$ is closed under complementation.
Second, an odd-length list can be padded to an odd upper bound by adding
equally many constant-zero and constant-one protocols. These operations
allow uniformly bounded indexed families of protocols to be given a
common number of guesses with cost $O(c+1)$ when their original
costs are at most $c$.

\begin{lemma}[Thresholds of accepting counts]
\label{lem:count-threshold}
Suppose a guess protocol has $\ell$ guesses and accepting count $a$.
For every integer $t$ with $1\le t\le\ell$, there is a PP protocol
with $2\ell+1$ guesses and gap
\[
2a-2t+1.
\]
It computes $\indic{a\ge t}$, and its complement computes
$\indic{a<t}$. Its deterministic constituents have the same maximum
cost as the original list, apart from constant protocols.
\end{lemma}

\begin{proof}
Append $\ell-t+1$ always-accepting and $t$ always-rejecting
guesses to the original list. The new gap is
$(2a-\ell)+(\ell-t+1)-t=2a-2t+1$.
It is never zero.
\end{proof}

\subsection{Signed sums and gap polynomials}

For a polynomial $p$ with integer coefficients, let $\normone p$
denote the sum of the absolute values of its coefficients.
This norm is subadditive and submultiplicative.

\begin{lemma}[Signed rectangle sums]
\label{lem:rectangle-sum}
Let $F(x,y)=\sum_j a_j R_j(x,y)$, where the $a_j$ are integers
and the $R_j$ are rectangle indicators. Put $W=\sum_j|a_j|$.
Then $\indic{F>0}$ has a PP protocol of cost
$\lceil\log(2W+1)\rceil+2$.
\end{lemma}

\begin{proof}
Let $K=\sum_{a_j<0}|a_j|$. Form a list with $a_j$ copies of
$R_j$ when $a_j>0$ and $|a_j|$ copies of $1-R_j$ when $a_j<0$.
It has $W$ guesses and accepting count $F+K$.
Append $W-K$ always-accepting and $K+1$ always-rejecting guesses.
The resulting gap is $2F-1$, which is positive exactly when the
integer $F$ is positive. The case $W=0$ gives a single rejecting guess.
\end{proof}

\begin{remark}[Charged trees and lists]
\label{rem:tree-list}
The list convention is equivalent, up to an additive constant in cost,
to the charged nondeterministic-tree convention of~\cite{DKMW04}.
A list becomes such a tree by first choosing its index and then running
the selected deterministic protocol. Unused binary indices can be
balanced by one accepting and one rejecting continuation.
Conversely, a depth-$t$ counting tree has at most $2^t$ leaves.
If $R_w$ is the rectangle of inputs consistent with a leaf transcript,
its gap is
\[
\sum_{w\ \mathrm{accepting}}R_w-\sum_{w\ \mathrm{rejecting}}R_w.
\]
Lemma~\ref{lem:rectangle-sum} realizes the same PP decision at list
cost $t+O(1)$. Enumerating accepting leaf rectangles likewise gives
the same accepting-count function class. This comparison concerns
charged counting trees, not an additive identity with every
randomized-protocol definition of weakly unbounded error.
\end{remark}

Concatenation of guess lists adds their gaps. Complementation negates
the gap. A product of two lists has one guess for each pair of guesses,
runs the two corresponding protocols, and accepts exactly when their
outputs agree. Its gap is the product of the two original gaps.

\begin{lemma}[Gap-polynomial realization]
\label{lem:gap-polynomial}
Let $\Pi_1,\ldots,\Pi_k$ have at most $\ell\ge1$ guesses each,
and suppose that every deterministic constituent costs at most $c$.
For every nonzero integer polynomial $H$ of total degree at most $d$,
there is a guess protocol whose gap is
$H(\gap_{\Pi_1},\ldots,\gap_{\Pi_k})$, with at most
$\normone H\ell^d$ guesses and cost at most
\[
\left\lceil\log\normone H+d\log\ell\right\rceil+dc.
\]
\end{lemma}

\begin{proof}
Realize each monomial by products of the corresponding lists, repeat
it according to the absolute value of its coefficient, and complement
when the coefficient is negative. Concatenate the resulting lists.
A monomial of degree $j$ uses at most $\ell^j$ guesses and at most
$jc$ communication per deterministic constituent. Summing these bounds
over the coefficients proves the claim.
\end{proof}

\begin{proposition}[Majority closure]
\label{prop:majority-closure}
If $f_1,\ldots,f_k$ have PP protocols of cost at most $c$, their
strict majority has a PP protocol of cost
\begin{equation}
\label{eq:majority-cost}
O\!\left(k(c+1)^2\log(k+1)\right).
\end{equation}
\end{proposition}

This is a quantitative form of the rational-approximation argument of
Beigel, Reingold, and Spielman~\cite{BRS95}.
A proof with explicit coefficient norms and positive denominators is
given in Appendix~\ref{app:majority}.
Communication PP intersection closure is also used in~\cite{GPW16}
and stated in~\cite{BMT21}; the bound here records the dependence
on the number and cost of the protocols.
In particular, polylogarithmically many efficient PP protocols can be
combined efficiently. The dependence on $k$ is important: this statement
does not imply closure under the quasipolynomially many inner protocols
allowed by a second majority quantifier.

\section{Bounded-error operators and perturbation measures}
\label{sec:bp}

\subsection{Language operators}

Fix a conventional, efficiently decodable pairing of a string and an
auxiliary string; the encoding increases length by at most the sum of
the two lengths plus logarithmic overhead. For a class $\mathcal C$,
a language $L$ belongs to $\BPdot\mathcal C$ if there are
$A\in\mathcal C$ and $q(n)\le\polylog(n)$ such that, for every
$x,y\in\B^n$,
\begin{align*}
(x,y)\in L&\ \Longrightarrow\
\Pr_{r\in\B^{q(n)}}[A(\langle x,r\rangle,\langle y,r\rangle)=1]
\ge\tfrac23,\\
(x,y)\notin L&\ \Longrightarrow\
\Pr_{r\in\B^{q(n)}}[A(\langle x,r\rangle,\langle y,r\rangle)=1]
\le\tfrac13.
\end{align*}
The auxiliary string is uniform and common to both parties.
This is the communication version of the bounded-error operator
used in~\cite{Wun12}.

The class $\RPdot\mathcal C$ is defined in the same way, with
acceptance probability zero on no-inputs and at least $2/3$ on
yes-inputs. We use the convention
\begin{equation}
\label{eq:zp-definition}
\ZPdot\mathcal C
:=(\RPdot\mathcal C)\cap\co(\RPdot\mathcal C),
\end{equation}
where $\co(\mathcal D)=\{\overline L:L\in\mathcal D\}$.
This definition fixes precisely the one-sided intersection meant by
the notation $\mathsf{ZP}$ in this paper.

\subsection{Randomized PP protocols and amplification}

A randomized PP protocol is a probability distribution over a finite
collection of PP protocols. Its cost is the largest cost in its support;
the randomness used to select a protocol is not charged at this stage.
Let $\RPP_\epsilon(f)$ denote the minimum such cost subject to
pointwise error at most $\epsilon$.
This is the PP instance of the randomized-counting framework
in~\cite{Wun12}. The amplification principle is the same; the
following bound uses the quantitative PP majority implementation.

\begin{proposition}[Error reduction]
\label{prop:amplification}
For $0\le\epsilon<1/2$ and every odd positive integer $k$,
\[
\RPP_{\frac12[4\epsilon(1-\epsilon)]^{k/2}}(f)
\le O\!\left(k(\RPP_\epsilon(f)+1)^2\log(k+1)\right).
\]
\end{proposition}

\begin{proof}
Choose $k$ protocols independently and take their majority using
Proposition~\ref{prop:majority-closure}. On a fixed input, the number
of errors is dominated by a binomial random variable with parameters
$k$ and $\epsilon$. For $j>k/2$,
\[
\epsilon^j(1-\epsilon)^{k-j}
\le[\epsilon(1-\epsilon)]^{k/2}.
\]
The binomial coefficients with $j>k/2$ sum to $2^{k-1}$ because
$k$ is odd. Summing gives the asserted error bound. The case
$\epsilon=0$ is immediate.
\end{proof}

We record the sampling step that bounds the common randomness.

\begin{lemma}[Finite-support sampling]
\label{lem:sampling}
Suppose a distribution on objects has a failure event of probability at
most $\eta$ for each of $N$ tests. For every $0<\delta<1$, there
is a multiset of $O(\delta^{-2}\log(N+1))$ objects on which the
uniform distribution has failure probability at most $\eta+\delta$
for every test. Its cardinality may be chosen to be a power of two.
Any property shared by every object in the original support is preserved.
\end{lemma}

\begin{proof}
Take $T$ independent samples. Hoeffding's inequality bounds the
probability that one test has empirical failure rate greater than
$\eta+\delta$ by $\exp(-2T\delta^2)$. A union bound over the
$N$ tests proves the claim when $T$ is a sufficiently large power of two.
\end{proof}

This is the standard sampling argument underlying Newman's
theorem~\cite{New91}.
The next proposition specializes the randomized-protocol/operator
correspondence of~\cite{Wun12} to PP.

\begin{proposition}
\label{prop:bp-protocols}
For every fixed $0<\epsilon<1/2$,
\[
\BPdot\PPcc
=\{L:\RPP_\epsilon(L_n)\le\polylog(n)\}.
\]
\end{proposition}

\begin{proof}
An indexed $\PPcc$ verifier gives a randomized PP protocol by
fixing its common random string. Constant-error amplification gives
the inclusion from left to right for any fixed $\epsilon>0$.

Conversely, first reduce the error to $1/6$ using
Proposition~\ref{prop:amplification}. Apply Lemma~\ref{lem:sampling}
with $N=2^{2n}$ and $\delta=1/6$. A uniform choice among $O(n+1)$
PP protocols now has error at most $1/3$ on every input, and its
index uses $O(\log(n+2))$ bits. The indexed predicate belongs to
$\PPcc$: the parties first verify that their auxiliary indices
agree, then use a common padded guess list for the selected protocol.
Both the index check and the PP cost are polylogarithmic.
\end{proof}

The exclusion of $\epsilon=0$ in this proposition is necessary.
With zero error, every protocol in the support computes the same
function exactly, so $\RPP_0(f)=\PPcost(f)$.
For fixed positive error, the same amplification and sampling argument
also permits error $2^{-(\log(n+2))^a}$ for any fixed $a>0$, while
preserving polylogarithmic communication and random-index length.

\subsection{Idempotence of bounded-error randomization}

The following is the PP specialization of the general BP-idempotence
principle in~\cite{Wun12}, once the required majority closure is
available. We include the argument to make the error and cost
conventions explicit.

\begin{proposition}
\label{prop:bp-idempotence}
\[
\BPdot(\BPdot\PPcc)=\BPdot\PPcc.
\]
Consequently, every fixed positive number of applications of
$\mathsf{BP}$ to $\PPcc$ gives the same class.
\end{proposition}

\begin{proof}
Suppose $L\in\BPdot(\BPdot\PPcc)$, with outer verifier
$A\in\BPdot\PPcc$ and outer error at most $1/3$.
By Proposition~\ref{prop:bp-protocols}, $A$ has randomized PP
protocols with error at most $1/12$ on every input and
polylogarithmic cost. On input $(x,y)$, choose the outer index $r$
and then independently choose such a protocol for
$A(\langle x,r\rangle,\langle y,r\rangle)$.
The paired input length is $n+\polylog(n)$, so the resulting PP
cost remains polylogarithmic in $n$.

For a fixed input, an error can occur only if the outer verifier's
Boolean answer is wrong or if the inner randomized PP protocol errs.
A union bound gives error at most $1/3+1/12=5/12<1/2$.
Proposition~\ref{prop:bp-protocols}, now applied with
$\epsilon=5/12$, places $L$ in $\BPdot\PPcc$.
The reverse inclusion follows by ignoring the outer random index.
Iteration proves the last assertion.
\end{proof}

The proof uses efficient amplification within $\PPcc$.
It does not assert idempotence of $\mathsf{BP}$ on an arbitrary
class lacking the required closure properties.

\subsection{The distributional perturbation operator}

Let $\Lambda$ be a real-valued measure on Boolean functions on a finite
domain $X\times Y$. Define
\begin{equation}
\label{eq:bp-measure}
(\mathsf{BP}_\epsilon\!\cdot\Lambda)(f)
:=\max_{\mu}\ 
\min_{g:\Pr_{z\sim\mu}[g(z)\ne f(z)]\le\epsilon}\Lambda(g),
\end{equation}
where $\mu$ ranges over all distributions on $X\times Y$.
The order of optimization is essential: the approximant may depend
on the input distribution, but a uniform cost bound is required for
every distribution.
This perturbation formalism and its finite-minimax interpretation
have modular-counting counterparts in~\cite{Wun12}.

\begin{proposition}[Minimax characterization]
\label{prop:minimax}
For every finite Boolean function $f$ and $0\le\epsilon<1/2$,
\[
\RPP_\epsilon(f)
=(\mathsf{BP}_\epsilon\!\cdot\PPcost)(f).
\]
\end{proposition}

\begin{proof}
For a cost bound $c$, let $\mathcal H_c$ be the finite set of
Boolean functions with PP complexity at most $c$. The finite minimax
theorem gives
\[
\min_{\lambda\in\Delta(\mathcal H_c)}\max_{z\in X\times Y}
\Pr_{h\sim\lambda}[h(z)\ne f(z)]
=\max_{\mu}\min_{h\in\mathcal H_c}
\Pr_{z\sim\mu}[h(z)\ne f(z)].
\]
Thus a randomized PP protocol of cost at most $c$ and pointwise
error at most $\epsilon$ exists if and only if every $\mu$ admits
an approximant of PP cost at most $c$ and $\mu$-error at most
$\epsilon$. Minimizing $c$ proves the proposition.
\end{proof}

\subsection{Discrepancy and margin complexity}

For a Boolean function $f$, write $S_f(x,y)=(-1)^{f(x,y)}$ for its
sign matrix. For a sign matrix $S$, define
\[
\disc(S)=\min_{\mu}\max_{A\times B}
\left|\sum_{(x,y)\in A\times B}\mu(x,y)S(x,y)\right|.
\]
Its margin complexity is
\begin{equation}
\label{eq:margin-definition}
\mc(S)=\inf\left\{
\max_x\|u_x\|_2\max_y\|v_y\|_2:
S(x,y)\langle u_x,v_y\rangle\ge1\text{ for all }x,y
\right\},
\end{equation}
where the vectors may have any finite common dimension.
Linial and Shraibman~\cite[Theorem 3.1]{LS09} prove
\begin{equation}
\label{eq:margin-discrepancy}
\frac{1}{8\disc(S)}\le\mc(S)\le\frac{1}{\disc(S)}.
\end{equation}
Together with the discrepancy characterization of weakly
unbounded-error complexity~\cite{Kla03,Kla07}, this gives, for
$f:\B^n\times\B^n\to\B$,
\begin{equation}
\label{eq:pp-margin}
\log\mc(S_f)-O(1)
\le\PPcost(f)
\le O\!\left(\log\mc(S_f)+\log(n+2)\right).
\end{equation}
For clarity, the lower bound also follows directly in our charged-list
model. After eliminating ties, expand every deterministic constituent
into its signed transcript rectangles. A protocol of cost $c$ gives
a correctly signed integer gap with total rectangle coefficient weight
$O(2^c)$. Correlation against a minimum-discrepancy distribution is
at least one in absolute value, so
$1\le O(2^c)\disc(S_f)$. Equation~\eqref{eq:margin-discrepancy}
implies $\log\mc(S_f)\le c+O(1)$.
The upper bound uses the standard discrepancy-to-protocol conversion;
the additive $O(\log(n+2))$ term allows its finite sampling overhead.

Define the distributional margin perturbation
\begin{equation}
\label{eq:distributional-margin}
\mathfrak m_\epsilon(f)
:=\max_{\mu}\ 
\min_{g:\Pr_{z\sim\mu}[g(z)\ne f(z)]\le\epsilon}\mc(S_g).
\end{equation}
Since logarithm is increasing,
$\log\mathfrak m_\epsilon(f)
=(\mathsf{BP}_\epsilon\!\cdot(\log\mc\circ S))(f)$.

\begin{theorem}[Distributional margin characterization]
\label{thm:margin-characterization}
For every fixed $0<\epsilon<1/2$,
\[
\BPdot\PPcc
=\left\{L:\log\mathfrak m_\epsilon(L_n)\le\polylog(n)\right\}.
\]
Equivalently, the raw parameter $\mathfrak m_\epsilon(L_n)$ is
bounded by $2^{\polylog(n)}$.
\end{theorem}

\begin{proof}
Apply the maximum and minimum in~\eqref{eq:bp-measure} to both
sides of~\eqref{eq:pp-margin}, and then use
Propositions~\ref{prop:bp-protocols} and~\ref{prop:minimax}.
The additive $O(\log(n+2))$ term does not affect efficient families.
\end{proof}

This is a \emph{distributional analogue} of the margin-complexity
rigidity considered in~\cite{LS09}. Their original parameter permits
a prescribed number of entry changes, or equivalently measures error
under the uniform distribution. Formula~\eqref{eq:distributional-margin}
instead takes a worst case over all input distributions.
The two definitions should not be identified. The distributional
formulation is the one characterized exactly by randomized protocols.
Theorem~\ref{thm:margin-characterization} is the precise PP/margin
counterpart of the modular-rank perturbation characterization
in~\cite{Wun12}, where the margin connection was also anticipated.

\section{Tarui's theorem in communication complexity}
\label{sec:tarui}

We use the following form of the probabilistic-polynomial construction.
The small coefficient norm is needed for communication complexity;
the Boolean guarantee yields one-sided error.

\begin{lemma}[Tarui-type probabilistic polynomials]
\label{lem:tarui-polynomials}
Fix a circuit depth $d$ and a constant $0<\eta<1/2$.
A Boolean circuit $C$ of size at most $S\ge2$ and depth at most $d$
has a distribution on integer polynomials $p$ such that
\begin{align}
\deg p&\le(\log(S+2))^{O(d)},
&\log(\normone p+1)&\le(\log(S+2))^{O(d)},
\label{eq:tarui-resources}\\
\Pr_p[p(z)=C(z)]&\ge1-\eta
&&\text{for every Boolean }z,
\label{eq:tarui-accuracy}\\
p(z)\in\B&\ \Longrightarrow\ p(z)=C(z)
&&\text{for every }p\text{ in the support}.
\label{eq:boolean-guarantee}
\end{align}
The constants in the bounds may depend on $d$ and $\eta$.
\end{lemma}

The construction is due to Tarui~\cite{Tar93}. The pointwise Boolean
guarantee is given explicitly by Hansen~\cite[Theorem~5]{Han08};
coefficient control and the rectangle expansion are also discussed
by Lokam~\cite[Theorem~5.19 and its proof]{Lok09}. For completeness,
Appendix~\ref{app:tarui} gives a proof of the precise properties used
here, including the coefficient bound. There is no need to impose a
small random seed before applying finite-support sampling.

\begin{theorem}[One-sided communication analogue of Tarui's theorem]
\label{thm:tarui}
\[
\PHcc\subseteq\ZPdot\PPcc\subseteq\BPdot\PPcc.
\]
\end{theorem}

\begin{proof}
Let $L\in\PHcc$. By the rectangle-circuit characterization,
for each $n$ there are rectangle indicators
$R_1(x,y),\ldots,R_{t_n}(x,y)$ and a constant-depth circuit $C_n$
of size $S_n\le2^{\polylog(n)}$ such that
\[
L_n(x,y)=C_n(R_1(x,y),\ldots,R_{t_n}(x,y)).
\]
Apply Lemma~\ref{lem:tarui-polynomials} with error $\eta=1/10$.
For a polynomial $p$ in its support, put
$P(x,y)=p(R_1(x,y),\ldots,R_{t_n}(x,y))$.
Products of rectangle indicators are rectangle indicators.
Consequently, $P$ is an integer linear combination of rectangles
with total absolute coefficient weight at most $\normone p$.
In particular, this weight is $2^{\polylog(n)}$.

For $b\in\B$, form the integer polynomial expression
\begin{equation}
\label{eq:equality-test}
E_b(x,y)=1-2(P(x,y)-b)^2.
\end{equation}
Its coefficient weight as a rectangle sum is at most
$1+2(\normone p+1)^2=2^{\polylog(n)}$.
Since $P$ is integer-valued,
\[
E_b(x,y)>0\quad\Longleftrightarrow\quad P(x,y)=b.
\]
Lemma~\ref{lem:rectangle-sum} therefore supplies an efficient
PP protocol for each equality test.

The Boolean guarantee implies that $P=1$ never occurs on a no-input,
and $P=0$ never occurs on a yes-input. The appropriate equality occurs
with probability at least $9/10$ on every input. Apply
Lemma~\ref{lem:sampling} simultaneously to the failure tests on all
$2^{2n}$ input pairs. A fixed power-of-two multiset of $O(n+1)$
polynomials preserves success probability at least $2/3$ everywhere.
Every polynomial in the multiset still has the Boolean guarantee.
The same common index, of length $O(\log(n+2))$, can select either
equality-test protocol. As in Proposition~\ref{prop:bp-protocols},
the indexed predicates belong to $\PPcc$.

Thus $L\in\RPdot\PPcc$ and
$\overline L\in\RPdot\PPcc$. Definition~\eqref{eq:zp-definition}
gives the first inclusion. The second follows because one-sided error
is a special case of bounded two-sided error.
\end{proof}

\begin{remark}
The Boolean guarantee is stronger than merely saying that a polynomial
is Boolean-valued with high probability. It forbids the \emph{wrong}
Boolean value for every random choice. This is exactly what makes
the equality tests in~\eqref{eq:equality-test} one-sided.
Under definition~\eqref{eq:zp-definition}, it resolves affirmatively
the one-sided-error question posed in the first version of this note.
The refinement is a consequence of Tarui's classical construction;
no new probabilistic-polynomial theorem is required.
\end{remark}

\begin{corollary}
\label{cor:ph-margin}
If a language $L$ has
$\log\mathfrak m_\epsilon(L_n)$ not bounded by any polylogarithm
for some fixed $0<\epsilon<1/2$, then $L\notin\PHcc$.
\end{corollary}

\begin{proof}
Combine Theorems~\ref{thm:margin-characterization} and~\ref{thm:tarui}.
\end{proof}

More generally, any lower bound for a measure that lower-bounds
$\PPcost$ transfers, via~\eqref{eq:bp-measure}, to a lower bound
for $\BPdot\PPcc$ and hence for $\PHcc$. This includes the
logarithmic forms of inverse discrepancy and margin complexity.
The general distributional-rigidity transfer strategy is already
present in~\cite{Wun12}; here the relevant target is randomized PP
and the matrix parameter is logarithmic margin complexity.

\subsection{A proper closure and an explicit margin gap}

Consider the family
\begin{equation}
\label{eq:rs-function}
F_m(x,y)=
\bigwedge_{i=1}^{m}\ \bigvee_{j=1}^{m^2}(x_{ij}\wedge y_{ij}),
\end{equation}
where each player receives $n=m^3$ bits. Ignoring unused coordinates
extends the family to all input lengths. The quantifiers
\emph{for every $i$ there exists $j$} give a $\Pi_2^{\mathrm{cc}}$
protocol of cost $O(\log(m+2))$.
Razborov and Sherstov~\cite[Theorem~1.1 and Corollary~1.2]{RS10}
proved
\begin{equation}
\label{eq:rs-lower}
\UPPcost(F_m)=\Omega(m).
\end{equation}

\begin{corollary}[Properness of the one-sided and bounded-error closures]
\label{cor:proper-closure}
\[
\PPcc\subsetneq\ZPdot\PPcc\subseteq\BPdot\PPcc,
\qquad
\ZPdot\PPcc\not\subseteq\UPPcc.
\]
In particular, $\BPdot\PPcc$ is a proper, idempotent
bounded-error closure of $\PPcc$ and is not contained in $\UPPcc$.
\end{corollary}

\begin{proof}
Since $\PPcc$ is closed under complementation, ignoring the random
index gives both $L\in\RPdot\PPcc$ and
$\overline L\in\RPdot\PPcc$ for every $L\in\PPcc$.
Thus $\PPcc\subseteq\ZPdot\PPcc$.
The family~\eqref{eq:rs-function} belongs to $\PHcc$, and hence
to $\ZPdot\PPcc$ by Theorem~\ref{thm:tarui}, but not to
$\UPPcc$ by~\eqref{eq:rs-lower}. Apply~\eqref{eq:upp-pp}
and Proposition~\ref{prop:bp-idempotence}.
\end{proof}

\begin{corollary}[Exact versus distributionally perturbed margin]
\label{cor:margin-gap}
For the family in~\eqref{eq:rs-function},
\[
\log\mc(S_{F_m})=\Omega(m),
\qquad
\log\mathfrak m_{1/3}(F_m)
\le(\log(m+2))^{O(1)}.
\]
\end{corollary}

\begin{proof}
Equations~\eqref{eq:upp-pp}, \eqref{eq:pp-margin},
and~\eqref{eq:rs-lower} give
\[
\Omega(m)\le\PPcost(F_m)+O(1)
\le O\!\left(\log\mc(S_{F_m})+\log(m+2)\right).
\]
This proves the first bound. The second follows from
$F\in\BPdot\PPcc$ and
Theorem~\ref{thm:margin-characterization}.
\end{proof}

Thus exact logarithmic margin complexity can be polynomially large in
the input length while its distributional perturbation is only
polylogarithmic. For every input distribution there is an appropriate
low-margin-complexity approximant; the approximant may depend on that
distribution. This statement does not provide a single approximant
that works for all distributions. Both corollaries are consequences
of the known lower bound~\cite{RS10} and the simulations above.

\subsection{A consequence for probabilistic margin complexity}
\label{sec:learning}

The distributional perturbation also gives a direct bound on the
probabilistic margin complexity introduced by Kamath, Montasser, and
Srebro~\cite[Definition~4]{KMS20}. We spell out the quantifiers, since
the feature-map distribution must be independent of the distribution
of examples.

For a finite class $\mathcal H\subseteq\{-1,1\}^{X}$, let
$K_\epsilon(\mathcal H)$ be the infimum of the radii $R\ge0$ for
which there is a distribution $\mathcal P$ over feature maps
$\phi:X\to\mathcal V$ into a common Hilbert space, with
$\|\phi(x)\|_2\le1$ for every $x$, such that
\begin{equation}
\label{eq:probabilistic-margin}
\mathbb E_{\phi\sim\mathcal P}
\inf_{\|w\|_2\le R}
\mathbb E_{x\sim D}
\indic{h(x)\langle w,\phi(x)\rangle\le1}
\le\epsilon
\end{equation}
for every distribution $D$ on $X$ and every $h\in\mathcal H$.
Thus $K_\epsilon$ is the distribution-independent version of
probabilistic margin complexity for bounded unit-margin loss.
The infimum convention avoids an immaterial boundary issue at
margin one. The exact margin complexity $\mc(\mathcal H)$ is the
margin complexity of its evaluation sign matrix.

\begin{proposition}[From perturbations to random features]
\label{prop:learning-bridge}
For a Boolean function $f:X\times Y\to\B$ on finite nonempty sets,
put $\mathcal H_f=\{x\mapsto S_f(x,y):y\in Y\}$. Then, for
$0\le\epsilon<1/2$,
\[
K_\epsilon(\mathcal H_f)\le2\mathfrak m_\epsilon(f).
\]
\end{proposition}

\begin{proof}
Put $M=\mathfrak m_\epsilon(f)$ and let $\mathcal G$ be the finite
collection of sign matrices $G$ with $\mc(G)\le M$.
Every distribution on $X\times Y$ admits a matrix in $\mathcal G$
that disagrees with $S_f$ with probability at most $\epsilon$.
Finite minimax, as in Proposition~\ref{prop:minimax}, gives a single
distribution on $\mathcal G$ such that
\[
\Pr_G[G(x,y)\ne S_f(x,y)]\le\epsilon
\quad\text{for every }(x,y)\in X\times Y.
\]

Fix $\eta>0$. For each $G$, choose factors $u_{G,x},v_{G,y}$,
both of norm at most one, such that
\[
G(x,y)\langle u_{G,x},v_{G,y}\rangle\ge\frac1{M+\eta}.
\]
Such factors exist by~\eqref{eq:margin-definition}.
There are only finitely many matrices, so zero padding places all
factors in one finite-dimensional Euclidean space.
Set $\phi_G(x)=u_{G,x}$ and
$w_{G,y}=2(M+\eta)v_{G,y}$.
Whenever $G(x,y)=S_f(x,y)$, the signed prediction is at least two.
Consequently,
\[
\mathbb E_G
\indic{S_f(x,y)\langle w_{G,y},\phi_G(x)\rangle\le1}
\le\epsilon
\quad\text{for every }x,y.
\]
Averaging over an arbitrary $D$ and allowing the infimum over
weights proves~\eqref{eq:probabilistic-margin} with radius
$2(M+\eta)$. The law of $\phi_G$ depends on neither $D$ nor the
target column. If different matrices give the same feature map,
the infimum is still at most the loss of each displayed witness.
Letting $\eta$ tend to zero proves the bound.
\end{proof}

\begin{corollary}[A finite-class margin separation]
\label{cor:learning-gap}
Let $F_m$ be the family in~\eqref{eq:rs-function} and
$\mathcal H_m=\mathcal H_{F_m}$. For every fixed
$0<\epsilon<1/2$,
\[
K_\epsilon(\mathcal H_m)
\le2^{(\log(m+2))^{O(1)}},
\qquad
\mc(\mathcal H_m)\ge2^{\Omega(m)}.
\]
The constant in the upper exponent may depend on $\epsilon$.
\end{corollary}

\begin{proof}
Theorem~\ref{thm:tarui} puts $F$ in $\BPdot\PPcc$.
Theorem~\ref{thm:margin-characterization} bounds
$\mathfrak m_\epsilon(F_m)$ quasipolynomially for every fixed
positive $\epsilon<1/2$. Apply
Proposition~\ref{prop:learning-bridge} for the first assertion.
Deleting repeated columns does not change margin complexity:
their representing vectors can be copied back. Therefore
$\mc(\mathcal H_m)=\mc(S_{F_m})$, and the lower bound is
Corollary~\ref{cor:margin-gap}.
\end{proof}

This is a consequence for representation complexity under bounded
margin loss. It asserts neither a hinge- or squared-loss bound nor
an efficient feature-computation or learning algorithm.
It also does not give a bounded probabilistic margin complexity for
a single infinite class. No optimality of the exponent is claimed.

\section{The majority operator and inner product}
\label{sec:pp-operator}

\subsection{Definition and normalization}

\begin{definition}[The majority operator]
\label{def:pp-operator}
A language $L$ belongs to $\PPdot\mathcal C$ if there are
$A\in\mathcal C$ and $q(n)\le\polylog(n)$ such that
\begin{equation}
\label{eq:pp-operator}
(x,y)\in L
\quad\Longleftrightarrow\quad
\left|\left\{r\in\B^{q(n)}:
A(\langle x,r\rangle,\langle y,r\rangle)=1\right\}\right|
>2^{q(n)-1}.
\end{equation}
\end{definition}

As with PP protocols, equality is a rejecting outcome. A second
application of this operator takes a majority of \emph{Boolean
decisions} made by inner majority protocols. It does not average their
acceptance probabilities or their unnormalized gaps.

\begin{lemma}[Power-of-two padding]
\label{lem:dyadic-padding}
A strict majority of an odd list of $m$ Boolean predicates can be
represented by a uniform common index of length
$\lceil\log(2m)\rceil$. Only duplication and constant predicates
are needed, and the resulting majority has no ties.
\end{lemma}

\begin{proof}
Duplicate every predicate and let
$M=2^{\lceil\log(2m)\rceil}$. Add $(M-2m)/2$ constant-one
predicates and the same number of constant-zero predicates.
If the original number of ones is $a$, the new number is
$2a+(M-2m)/2$, which exceeds $M/2$ exactly when $2a>m$.
Since $m$ is odd, equality is impossible.
\end{proof}

\begin{proposition}
\label{prop:pp-base}
\[
\PPdot\Pcc=\PPcc.
\]
\end{proposition}

\begin{proof}
Fixing the common index in a $\Pcc$ verifier gives a list of
$2^{\polylog(n)}$ deterministic protocols of polylogarithmic cost.
Its PP cost is therefore polylogarithmic.

Conversely, start with an efficient PP guess protocol. Duplicate its
list and append one rejecting guess to make its length odd and its
gap nonzero. Lemma~\ref{lem:dyadic-padding} then gives a
power-of-two list with the same output and only constant additive
overhead in the logarithm of its length. Given the index, the parties
run the corresponding deterministic protocol. The indexed verifier
belongs to $\Pcc$; unequal or malformed auxiliary strings are
rejected after a polylogarithmic-cost check.
\end{proof}

\subsection{Sasaki's formula and inner product}

For odd $N$, let $\Maj_N$ denote strict majority on $N$ bits.
The following cyclic-majority identity is attributed to Sasaki in
Knuth~\cite[Section~7.1.1]{Knu11}.
We include a proof to make the construction self-contained.

\begin{lemma}[Sasaki's cyclic-majority identity]
\label{lem:sasaki}
Let $m\ge1$ and $z_0,z_1,\ldots,z_{2m}\in\B$.
Interpret the indices $1,\ldots,2m$ cyclically, with $z_0$ kept
separate. For $1\le j\le2m$, set
\begin{equation}
\label{eq:sasaki-inner}
s_j=\Maj_{2m+1}\bigl(
z_0,z_j,\ldots,z_{j+m-1},
1-z_{j+m},\ldots,1-z_{j+2m-1}\bigr).
\end{equation}
Then
\begin{equation}
\label{eq:sasaki}
z_0\oplus z_1\oplus\cdots\oplus z_{2m}
=\Maj_{2m+1}(1-z_0,s_1,\ldots,s_{2m}).
\end{equation}
\end{lemma}

\begin{proof}
Write $k=\sum_{i=1}^{2m}z_i$ and
$a_j=\sum_{r=0}^{m-1}z_{j+r}$. Opposite windows satisfy
$a_{j+m}=k-a_j$, and
$s_j=\indic{z_0+2a_j-k>0}$. Therefore
\[
s_j+s_{j+m}=1+(2z_0-1)\indic{2a_j=k}.
\]
Put $T=|\{j\in[m]:2a_j=k\}|$. The number of ones in the
outer majority is
\[
(1-z_0)+\sum_{j=1}^{2m}s_j
=m+1-z_0+(2z_0-1)T.
\]
If $k$ is odd, then $T=0$, and the outer majority is $1-z_0$.
If $k$ is even, then $T\ge1$: the window counts $a_1$ and
$a_{m+1}=k-a_1$ lie on opposite sides of $k/2$, while successive
window counts differ by at most one. Some window therefore has
exactly $k/2$ ones; its opposite window has the same count.
In this case the outer majority is $z_0$.
In both cases the output is $z_0\oplus(k\bmod2)$.
\end{proof}

\begin{theorem}[Inner product]
\label{thm:inner-product}
The inner-product language belongs to $\PPdot\PPcc$.
More precisely, the outer common index has length $O(\log(n+2))$,
and every indexed inner PP protocol has cost $O(\log(n+2))$.
\end{theorem}

\begin{proof}
The cases $n\le1$ are immediate. Otherwise, pad the $n$
coordinate products $x_i y_i$ with at most one zero to obtain an
odd-length list $z_0,\ldots,z_{2m}$, where $2m+1\le n+1$.
Each $z_i$, and each complement $1-z_i$, is a deterministic
predicate of communication cost at most two.

Each inner majority $s_j$ in~\eqref{eq:sasaki-inner} is therefore
a PP guess protocol with $2m+1$ guesses and cost at most
$\lceil\log(2m+1)\rceil+2$. The remaining outer predicate
$1-z_0$ is deterministic and can be given the same number of guesses
by repetition. Lemma~\ref{lem:sasaki} shows that the strict majority
of these $2m+1$ predicates is $\IP_n$.

Apply Lemma~\ref{lem:dyadic-padding} to the outer list.
A single indexed verifier handles the cyclic choice $j$, the predicate
$1-z_0$, and the padding constants. Checking that the two auxiliary
indices agree costs $O(\log(n+2))$. The logarithm of the common inner
guess count and each constituent communication cost are both
$O(\log(n+2))$. Thus the indexed verifier belongs to $\PPcc$, and the
outer index has length $O(\log(n+2))$.
\end{proof}

\begin{corollary}[Parity acceptance]
\label{cor:parity-class}
\[
\ParityPcc\subseteq\PPdot\PPcc.
\]
\end{corollary}

\begin{proof}
Replace the coordinate tests in the preceding proof by the deterministic
constituents of an efficient guess protocol. Its number of guesses is
$\ell_n\le2^{\polylog(n)}$. Pad with rejecting guesses to an odd
length at least three, without changing the accepting count. Sasaki's
formula computes the parity of that count. The common outer index has
length $O(\log(\ell_n+2))$, and each inner protocol has
polylogarithmic cost.
\end{proof}

\begin{remark}[An alternative accepting-count identity]
There is also a direct threshold-count proof. For odd $N$ and
$a=\sum_{i=1}^{N}z_i$, put $h_t=\indic{a\ge t}$ for odd $t$
and $h_t=\indic{a<t}$ for even $t$. Then
\[
\sum_{t=1}^{N}h_t
=\frac{N-1}{2}+\sum_{t=1}^{a}(-1)^{t+1}
=\frac{N-1}{2}+(a\bmod2).
\]
The majority of the $h_t$ is therefore parity.
Lemma~\ref{lem:count-threshold} supplies the corresponding inner
PP protocols. This is a separate representation from Sasaki's
cyclic-majority formula.
\end{remark}

\begin{corollary}[Non-idempotence and separation from unrestricted error]
\label{cor:strict}
\[
\PPcc\subsetneq\PPdot\PPcc,
\qquad
\PPdot\PPcc\not\subseteq\UPPcc.
\]
Consequently, the majority operator is not idempotent even on $\Pcc$:
\[
\PPdot(\PPdot\Pcc)\ne\PPdot\Pcc.
\]
\end{corollary}

\begin{proof}
Every language in $\PPcc$ belongs to $\PPdot\PPcc$ by using an
indexed predicate that ignores its index. Forster's theorem~\cite{For02}
gives unbounded-error communication complexity $\Omega(n)$ for
$\IP_n$. Hence $\IP\notin\UPPcc$, and a fortiori
$\IP\notin\PPcc$. Apply Theorem~\ref{thm:inner-product} and
Proposition~\ref{prop:pp-base}.
\end{proof}

The weaker exclusion $\IP\notin\PPcc$ also follows from the
modular-orthogonality lower bound of~\cite{DKMW04}, specialized to
modulus two and complemented. Forster's theorem is used here for
the stronger separation from $\UPPcc$.

The minimum outer advantage in the inner-product construction is
$\Theta(1/n)$, not a constant. The argument therefore does not place
$\IP$ in $\BPdot\PPcc$. Flattening the two levels by sampling
one branch from each would average the inner biases rather than their
Boolean signs, and need not preserve the answer.
The order of the advantage is tight for this representation:
setting exactly one non-special coordinate product to one gives an
outer majority gap of one before padding, and dyadic padding changes
the advantage only by a constant factor.

\subsection{Arbitrary predicates of accepting counts}

The accepting-count identity above extends from parity to every
Boolean predicate of a count. The underlying representation of a
symmetric Boolean function by two levels of threshold gates is
classical; see Knuth~\cite[Section~7.1.1]{Knu11}, who
attributes the general construction to Minnick, and
Siu, Roychowdhury, and Kailath~\cite[Lemma~1]{SRK91}.
We give a transition-count formulation of this threshold-telescoping
construction that makes the communication costs explicit.

\begin{theorem}[Boolean predicates of accepting counts]
\label{thm:count-predicates}
Let $\Pi$ have $\ell\ge1$ guesses, with accepting count
$a=\acc_\Pi$, and suppose each deterministic constituent has cost
at most $c$. Let $b:\{0,\ldots,\ell\}\to\B$ and put
\[
T=\{t\in\{1,\ldots,\ell\}:b(t)\ne b(t-1)\},
\qquad r=|T|.
\]
Then $b(a)$ is a strict majority of an odd number $M\le r+1$
of predicates, each of which is constant or has a PP protocol of
cost at most $\lceil\log(2\ell+1)\rceil+c$.
The outer majority can be selected by a uniform common index of
length $O(\log(r+2))$.
Consequently, every Boolean predicate of an efficient accepting-count
family belongs to $\PPdot\PPcc$.
\end{theorem}

\begin{proof}
For each $t\in T$, define
\[
h_t=
\begin{cases}
\indic{a\ge t},& b(t-1)=0,\ b(t)=1,\\
\indic{a<t},& b(t-1)=1,\ b(t)=0.
\end{cases}
\]
Let $d$ be the number of downward transitions. Telescoping gives
\begin{equation}
\label{eq:transition-sum}
\sum_{t\in T}h_t
=d+\sum_{t\in T}(b(t)-b(t-1))\indic{a\ge t}
=d+b(a)-b(0).
\end{equation}
If $b(0)=b(\ell)$, append one constant predicate with value
$b(0)$; otherwise append nothing. The resulting list has odd
length $M$: it is $r+1$ when the endpoint values agree and $r$
when they differ. Equation~\eqref{eq:transition-sum} now says
that its number of ones is exactly
\[
\frac{M-1}{2}+b(a).
\]
This also covers $r=0$, when the list consists of one constant.
Its strict majority is therefore $b(a)$, with no ties.

Lemma~\ref{lem:count-threshold} realizes each nonconstant $h_t$
with $2\ell+1$ guesses. Apply
Lemma~\ref{lem:dyadic-padding} to the outer list.
For a family of efficient guess protocols, both the logarithm of
$\ell$ and $c$ are polylogarithmic. A common indexed verifier
selects the threshold and its polarity, or a padding constant.
The parties check their auxiliary indices and encodings and use a
common padded inner list. Its total cost is polylogarithmic, as
required. The values of $b$ and the index-to-threshold map are
part of the nonuniform protocol description.
\end{proof}

In particular, the theorem applies to any specified residue or set
of residues of an accepting count, not just to parity. No uniform
computability condition on the predicates $b$ is needed in the
nonuniform model used here.
The generalized-inner-product completeness theorem of~\cite{DKMW04}
provides related context for arbitrary acceptance modes. It preserves
the chosen acceptance rule; the construction above instead implements
that rule by two explicit majority layers.

\begin{corollary}[Symmetric coordinate counts]
\label{cor:symmetric-counts}
For arbitrary predicates $D_n:\{0,\ldots,n\}\to\B$, both
families
\[
f_{D_n}(x,y)=D_n\!\left(\sum_{i=1}^{n}x_i y_i\right),
\qquad
g_{D_n}(x,y)=D_n\!\left(\sum_{i=1}^{n}(x_i\oplus y_i)\right)
\]
belong to $\PPdot\PPcc$.
In both cases, the outer index length and the cost of the indexed
inner PP verifier are $O(\log(n+2))$.
\end{corollary}

\begin{proof}
Use one guess per coordinate, testing respectively $x_i y_i$ or
$x_i\oplus y_i$ with at most two bits of communication.
Apply Theorem~\ref{thm:count-predicates}, with $\ell=n$ and
$r\le n$. The case $n=0$ is constant.
\end{proof}

\subsection{Few transitions and symmetric set intersection}

A small number of changes permits a one-level PP protocol as well.

\begin{proposition}[A product-of-thresholds bound]
\label{prop:few-transitions}
Under the assumptions of Theorem~\ref{thm:count-predicates},
\[
\PPcost(b\circ\acc_\Pi)
\le r\bigl(\lceil\log(2\ell+1)\rceil+c\bigr)+O(1).
\]
\end{proposition}

\begin{proof}
For integer $a\in\{0,\ldots,\ell\}$, consider
\begin{equation}
\label{eq:transition-product}
G(a)=(2b(0)-1)(-1)^r
\prod_{t\in T}(2a-2t+1).
\end{equation}
None of the factors vanishes at an integer. The sign of $G(0)$
is $2b(0)-1$, and its sign changes between $a-1$ and $a$
exactly when $a\in T$. Hence $G(a)>0$ if and only if $b(a)=1$.

By Lemma~\ref{lem:count-threshold}, each factor in
\eqref{eq:transition-product} is the gap of a list with
$2\ell+1$ guesses and constituent cost at most $c$.
Taking products of these lists and, if necessary, complementing
the output realizes $G$ with $(2\ell+1)^r$ guesses and
constituent cost at most $rc$. This proves the bound.
For $r=0$, use the appropriate constant protocol.
\end{proof}

For the set-intersection family $f_{D_n}$ in
Corollary~\ref{cor:symmetric-counts}, write
\[
r_n=\bigl|\{t\in\{1,\ldots,n\}:D_n(t)\ne D_n(t-1)\}\bigr|.
\]
Sherstov~\cite[Theorem~1.1]{She11} proved
\begin{equation}
\label{eq:symmetric-upp-lower}
\UPPcost(f_{D_n})
=\Omega\!\left(\frac{r_n}{\log^5(n+2)}\right).
\end{equation}
Together with the elementary upper bound above, this yields the
following characterization.

\begin{corollary}[Class-level characterization for symmetric intersection]
\label{cor:symmetric-characterization}
For every sequence of predicates $D_n$,
\[
(f_{D_n})\in\PPcc
\quad\Longleftrightarrow\quad
(f_{D_n})\in\UPPcc
\quad\Longleftrightarrow\quad
r_n\le\polylog(n).
\]
Moreover, every such family lies in $\PPdot\PPcc$, regardless
of the number of transitions.
\end{corollary}

\begin{proof}
Proposition~\ref{prop:few-transitions} gives
$\PPcost(f_{D_n})=O((r_n+1)\log(n+2))$.
This proves the implication from a polylogarithmic transition count
to membership in $\PPcc$.
Equation~\eqref{eq:upp-pp} gives the next implication.
Finally,~\eqref{eq:symmetric-upp-lower} forces a polylogarithmic
transition count for every family in $\UPPcc$.
The last assertion is Corollary~\ref{cor:symmetric-counts}.
\end{proof}

For example, for any fixed $d\ge2$, the predicate
$D_n(s)=\indic{s\not\equiv0\pmod d}$ has $\Theta(n)$ transitions.
The resulting modular-intersection family has UPP complexity
$\Omega(n/\log^5(n+2))$, yet admits the logarithmic-resource
two-majority representation above. Thus inner product is one
member of a broader collection of explicit witnesses to
$\PPdot\PPcc\not\subseteq\UPPcc$.

The converse direction in
Corollary~\ref{cor:symmetric-characterization} is specific to
set intersection. It is not asserted for arbitrary accepting-count
protocols or for the Hamming-distance families $g_{D_n}$.
The general implication from few transitions to small PP cost
does hold in all of these settings.

\subsection{Outer advantage and bounded-error simulation}

The distinction between the two operators can also be quantified.

\begin{proposition}[An outer-advantage bound]
\label{prop:outer-advantage}
Let $f$ be a Boolean function, and suppose a distribution over a
finite family of Boolean predicates $A_r$ satisfies
\[
\PPcost(A_r)\le c\quad\hbox{for every }r,
\qquad
\Pr_r[A_r(x,y)=f(x,y)]\ge\tfrac12+\gamma
\quad\hbox{for every }(x,y),
\]
where $0<\gamma\le1/2$.
For every $0<\delta<1/2$ and any odd integer
$k\ge(2\gamma^2)^{-1}\ln(1/\delta)$,
\[
\RPP_\delta(f)
\le O\!\left(k(c+1)^2\log(k+1)\right).
\]
In particular,
\begin{equation}
\label{eq:outer-advantage}
\RPP_{1/3}(f)
\le O\!\left(\gamma^{-2}(c+1)^2\log(2/\gamma)\right).
\end{equation}
\end{proposition}

\begin{proof}
Choose $k$ indices independently and take the majority of the
corresponding Boolean predicates. For any fixed input, Hoeffding's
inequality bounds its error by
$\exp(-2k\gamma^2)\le\delta$.
For each fixed tuple of indices,
Proposition~\ref{prop:majority-closure} realizes this majority
as a PP protocol with the stated cost.
The distribution over tuples is therefore a randomized PP protocol.
For~\eqref{eq:outer-advantage}, take the smallest admissible
odd $k$ with $\delta=1/3$.
\end{proof}

\begin{corollary}
\label{cor:outer-advantage-class}
If a $\PPdot\PPcc$ representation has correct outer decision
probability at least $1/2+\gamma_n$ on every input, with
$\gamma_n^{-1}\le\polylog(n)$, then its language belongs to
$\BPdot\PPcc$.
\end{corollary}

\begin{proof}
Apply~\eqref{eq:outer-advantage} and
Proposition~\ref{prop:bp-protocols}.
\end{proof}

Thus an inverse-polylogarithmic outer advantage suffices for
bounded-error simulation with efficient cost. In the Sasaki
representation of inner product, the minimum outer advantage is
only $\Theta(1/n)$, so this estimate does not give an efficient
bounded-error protocol. No lower bound against
$\BPdot\PPcc$ for inner product follows from this observation.

\section{Upper bounds and remaining questions}
\label{sec:discussion}

\begin{proposition}
\label{prop:space-upper}
\[
\BPdot\PPcc\subseteq\PPdot\PPcc\subseteq\PSPACEcc.
\]
\end{proposition}

\begin{proof}
The first inclusion is immediate from the probability thresholds in
the definitions. For the second, expand a $\PPdot\PPcc$ protocol
into an outer majority of $2^{\polylog(n)}$ inner PP predicates,
each of which is a majority of $2^{\polylog(n)}$ deterministic
protocols of polylogarithmic cost. Padding makes these bounds uniform.

Majority on $N$ Boolean inputs has bounded-fan-in Boolean circuits
of depth $O(\log^2(N+2))$: add the bits in a balanced tree of
binary adders and compare the sum with $N/2$. Each deterministic
communication protocol can likewise be replaced by a Boolean formula
of polylogarithmic depth over local predicates, or by an OR of its
accepting transcript rectangles with a balanced OR tree.
Replacing both majority layers by these circuits yields a
bounded-fan-in Boolean circuit of polylogarithmic depth over rectangle
indicators. The alternating evaluation described in
Section~\ref{sec:preliminaries} places the language in $\PSPACEcc$.
\end{proof}

Writing $\mathcal M^{\mathrm{cc}}$ for the class in
Theorem~\ref{thm:margin-characterization}, we obtain the combined picture
\begin{align*}
\PHcc&\subseteq\ZPdot\PPcc
\subseteq\BPdot\PPcc=\mathcal M^{\mathrm{cc}},\\
\mathcal M^{\mathrm{cc}}&\subseteq\PPdot\PPcc
\subseteq\PSPACEcc,\\
\ParityPcc&\subseteq\PPdot\PPcc.
\end{align*}
In addition,
\[
\PPcc\subsetneq\ZPdot\PPcc,
\qquad
\BPdot\mathcal M^{\mathrm{cc}}=\mathcal M^{\mathrm{cc}},
\qquad
\mathcal M^{\mathrm{cc}}\not\subseteq\UPPcc.
\]
These inclusions clarify the scope of the rigidity approach.
The class defined by distributionally small margin complexity is
not larger than $\PSPACEcc$, but equality is not ruled out.
Nor do the preceding arguments establish a strict containment between
$\PHcc$ and $\mathcal M^{\mathrm{cc}}$.
They also do not establish strictness of
$\ZPdot\PPcc\subseteq\BPdot\PPcc$ or of
$\BPdot\PPcc\subseteq\PPdot\PPcc$.

The results also separate two distinct obstacles. Proving that the
inner-product family has large distributional margin perturbation
would give a lower bound against $\BPdot\PPcc$, and hence against
$\PHcc$. In contrast, inner product already belongs to the
two-majority class $\PPdot\PPcc$. Thus a lower bound for this
latter class requires a different example. More generally,
Corollary~\ref{cor:symmetric-counts} rules out every symmetric
predicate of set-intersection size or Hamming distance as such
a witness. A predicate of any efficient accepting count is
likewise excluded by Theorem~\ref{thm:count-predicates}.
The outer-advantage bound isolates one quantitative distinction:
an inverse-polylogarithmic advantage collapses a two-level
representation into $\BPdot\PPcc$, whereas a merely nonzero
advantage does not provide the same simulation bound.

Finally, the containment of $\UPPcc$ in $\PPdot\PPcc$ is not
settled here. Such a containment would imply
$\UPPcc\subseteq\PSPACEcc$, whereas the reverse containment
$\PPdot\PPcc\subseteq\UPPcc$ is excluded by
Corollary~\ref{cor:strict}. These statements concern the explicitly
defined dot operator and should not be replaced by assertions about
oracle classes without a separate simulation argument.

\appendix

\section{Quantitative majority closure}
\label{app:majority}

We prove Proposition~\ref{prop:majority-closure}. The argument uses
an elementary rational approximation to the sign function, following
the method of~\cite{BRS95}. Keeping coefficient \(\ell_1\)-norms
avoids losing factors when products and common denominators are formed.

For integers $m\ge0$ and $h\ge1$, define
\begin{align*}
Q_m(z)&=\prod_{j=0}^{m}(z-2^j)^2,\\
A_{m,h}(z)&=Q_m(-z)^h-Q_m(z)^h,\\
B_{m,h}(z)&=Q_m(-z)^h+Q_m(z)^h,
\qquad
R_{m,h}(z)=\frac{A_{m,h}(z)}{B_{m,h}(z)}.
\end{align*}
The denominator is positive for every real $z$: its two summands are
nonnegative and have disjoint zero sets.
For $1\le z\le2^m$, one of the factors in
\[
\rho(z)=\frac{Q_m(z)}{Q_m(-z)}
=\prod_{j=0}^m\left(\frac{z-2^j}{z+2^j}\right)^2
\]
is at most $1/9$, and every factor is at most one. Indeed, choose
$j=\lceil\log z\rceil$; then $2^{j-1}\le z\le2^j$, with
the endpoint $z=1$ giving a zero factor. Hence
\[
0\le1-R_{m,h}(z)
=\frac{2\rho(z)^h}{1+\rho(z)^h}
\le2\cdot9^{-h}.
\]
Since $R_{m,h}$ is odd, this proves
\begin{equation}
\label{eq:rational-error}
|R_{m,h}(z)-\sgn z|\le2\cdot9^{-h}
\qquad(1\le|z|\le2^m).
\end{equation}
Furthermore,
\[
\normone{Q_m}
\le\prod_{j=0}^{m}(1+2^j)^2
\le2^{(m+1)(m+2)}.
\]
Both $A_{m,h}$ and $B_{m,h}$ therefore have degree at most
$2h(m+1)$ and coefficient norm at most
\begin{equation}
\label{eq:rational-norm}
C_{m,h}=2^{1+h(m+1)(m+2)}.
\end{equation}

Now eliminate ties in the $k$ given PP protocols. With
$m=\lceil c\rceil+2$, their new gaps $g_i$ satisfy
$1\le|g_i|\le2^m$; each protocol uses at most $2^m$ guesses
and has deterministic constituent cost at most $m$.
Choose $h=\lceil\log_9(8k)\rceil$, and abbreviate
$A=A_{m,h}$, $B=B_{m,h}$, and $R=R_{m,h}$.
If exactly $p$ of the gaps are positive, then~\eqref{eq:rational-error}
implies
\[
\left|\sum_{i=1}^k R(g_i)-(2p-k)\right|\le\tfrac14.
\]
Consequently, $2\sum_iR(g_i)-1$ is positive when $p>k/2$
and negative otherwise, including a tie at the outer majority.
Clearing the positive denominators gives the integer polynomial
\[
H(z_1,\ldots,z_k)
=2\sum_{i=1}^k A(z_i)\prod_{j\ne i}B(z_j)
-\prod_{j=1}^k B(z_j).
\]
Its degree is at most $d=2kh(m+1)$ and
\[
\normone H\le(2k+1)C_{m,h}^k.
\]
Lemma~\ref{lem:gap-polynomial} realizes its sign with cost at most
\[
\left\lceil\log(2k+1)+k[1+h(m+1)(m+2)]+dm\right\rceil+dm
=O\!\left(k(c+1)^2\log(k+1)\right).
\]
This completes the proof of Proposition~\ref{prop:majority-closure}.

\section{Probabilistic polynomials with a Boolean guarantee}
\label{app:tarui}

We give the resource accounting for Lemma~\ref{lem:tarui-polynomials},
using the NOR construction in~\cite{Han08} and the
coefficient-norm viewpoint of~\cite{Lok09}.
Replacing AND and OR gates by NOR gates and inserting negations increases
depth by at most a constant factor and size by at most a constant factor.
It is therefore enough to consider a NOR circuit of constant depth $d$,
size at most $S\ge2$, and fan-in at most $S$.
Set
\[
L=\lceil\log S\rceil,
\qquad r=\lceil8\ln(S/\eta)\rceil,
\qquad D=1+2r(L+1).
\]
At an input gate use its input variable or its Boolean constant.
At a NOR gate with child polynomials $q_1,\ldots,q_k$, use
\begin{equation}
\label{eq:tarui-nor}
p=\left(1+\sum_{i=1}^kq_i\right)
\prod_{\ell=1}^{r}\prod_{j=0}^{L}
\left(1-\sum_{i=1}^k a_{i,j,\ell}q_i\right)^2,
\end{equation}
where the coefficients $a_{i,j,\ell}$ are independent Bernoulli
variables with mean $2^{-j}$, with independent choices at distinct gates.

\paragraph{Boolean guarantee.}
For every Boolean external input, all polynomials in the construction
take nonnegative integer values. Suppose inductively that each child
has the Boolean guarantee. If $p=1$, the first factor in
\eqref{eq:tarui-nor} must be one, so every $q_i=0$.
The true NOR output is therefore one. If $p=0$, some squared factor
vanishes, so a selected sum of nonnegative integers equals one.
At least one selected child has value $q_i=1$, and its true Boolean
value is one. The true NOR output is therefore zero. Induction proves
\eqref{eq:boolean-guarantee} for every random choice.

\paragraph{Accuracy.}
Condition on the child polynomials taking their correct Boolean values.
If all child values are zero, the gate polynomial is one deterministically.
Otherwise let $t\ge1$ be the number of one-valued children and choose
$j=\lceil\log t\rceil$. At this scale a random subset contains
exactly one of these children with probability at least $1/8$.
For $t=1$ this follows from $j=0$; for $t\ge2$, it follows from
$1/(2t)\le2^{-j}\le1/t$ and
$(1-1/t)^{t-1}\ge e^{-1}$.
If such a selection occurs, one factor in~\eqref{eq:tarui-nor}
vanishes and the gate output is correctly zero. The probability of
missing in all $r$ repetitions is at most
$(7/8)^r\le\eta/S$.
Taking a union bound over the first erroneous gate in a topological
ordering gives total error at most $\eta$.

\paragraph{Degree and coefficient norm.}
Let $K_h$ bound the coefficient norm at depth at most $h$.
The construction gives degree at most $D^h$, and
\[
K_0=1,
\qquad K_h\le((S+1)K_{h-1})^D.
\]
Thus
\[
\log K_d
\le\log(S+1)\sum_{h=1}^dD^h
\le2D^d\log(S+1).
\]
Since $d$ and $\eta$ are fixed and $D=O(\log^2(S+2))$,
these are the bounds in~\eqref{eq:tarui-resources}.
Multilinearizing the polynomial on Boolean inputs can only decrease
its degree and coefficient norm. This proves
Lemma~\ref{lem:tarui-polynomials}.

\section*{Disclosure of AI assistance}
OpenAI ChatGPT was used as an assistance tool for drafting, checks, and locating related literature.  
The author reviewed, revised, and takes responsibility for the mathematical statements, arguments, citations, and final contents of the manuscript.

\bibliographystyle{alpha}
% \begingroup
% \raggedright
\bibliography{note_on_cc}
% \endgroup
\end{document}